\documentclass[12pt, a4paper, leqno]{amsart}
\usepackage{esint} 

\usepackage{amsthm}
\usepackage{amsmath}
\usepackage{amsfonts}
\usepackage{amssymb}
\usepackage{amscd}
\usepackage{mathrsfs}
\usepackage{enumerate}
\usepackage{bm}
\usepackage{dsfont}
\usepackage{xcolor}

\usepackage[utf8]{inputenc}
\usepackage[english]{babel}
\usepackage{anysize}

\usepackage{lipsum}

\newcommand{\R}{\mathbb{R}}
\newcommand{\Op}{\operatorname{Op}}
\newcommand{\Ad}{\operatorname{Ad}}

\newcommand{\supp}{\operatorname{supp}}

\newtheorem{prop}{Proposition}
\newtheorem{lemma}{Lemma}

\newtheorem{definition}{Definition}

\newtheorem{teor}{Theorem}
\theoremstyle{remark}

\theoremstyle{remark}

\theoremstyle{remark}
\newtheorem{remark}{Remark}

\title[Semiclassical asymptotics for nonselfadjoint Harmonic oscillators]{Semiclassical asymptotics for nonselfadjoint Harmonic oscillators}
\author[Víctor Arnaiz]{Víctor Arnaiz}
 
\address{Instituto de Ciencias Matemáticas (ICMAT), C/ Nicolás Cabrera, nº 13-15, Campus de Cantoblanco, UAM. 28049 Madrid, España}

\email{victor.arnaiz@icmat.es}

\author[Gabriel Rivi\`ere]{Gabriel Rivi\`ere}

\address{Laboratoire de math\'ematiques Jean Leray (U.M.R. CNRS 6629), Universit\'e de Nantes, 2 rue de la Houssini\`ere, BP92208, 
44322 Nantes Cedex 3, France}

\email{gabriel.riviere@univ-nantes.fr}

\begin{document}

\begin{abstract}
We consider nonselfadjoint perturbations of semiclassical harmonic oscillators. Under appropriate dynamical assumptions, we establish some 
spectral estimates such as upper bounds on the resolvent near the real axis when no geometric control condition is satisfied. 
\end{abstract}


\maketitle

\section{Introduction}

Motivated by earlier works of Lebeau on the asymptotic properties of the damped wave equation~\cite{Lebeau96}, Sj\"ostrand initiated in~\cite{Sjostrand00} the spectral study of this partial differential equation 
on compact Riemannian manifolds. He proved that eigenfrequencies verify a Weyl asymptotics in the high frequency limit~\cite[Th.~0.1]{Sjostrand00} -- see also~\cite{Markus88,Markus_Matsaev80} for earlier related contributions of Markus and Matsaev. Moreover, he showed that eigenfrequencies lie in a strip of the complex plane which can be completely determined in terms of the average of the damping function along the geodesic flow~\cite[Th.~0.0 and 0.2]{Sjostrand00} -- see also~\cite{Lebeau96, RauchTaylor}. Following~\cite{Sjostrand00}, showing these results turns out to be the particular case of a more systematic study of a nonselfadjoint semiclassical problem which has since then been the object of several works. More precisely, it was investigated how these generalized eigenvalues are asymptotically distributed inside the strip determined by Sj\"ostrand and how the dynamics of the underlying classical Hamiltonian influences this asymptotic distribution. Mostly two questions have been considered in the literature. First, one can ask about the precise distribution of eigenvalues inside the strip and this question was addressed both in the completely integrable framework~\cite{Hitrik02, HitrikSjostrand04, HitrikSjostrand05, HitrikSjostrandVuNgoc07, HitrikSjostrand08, HitrikSjostrand08b, HitrikSjostrand12, HitrikSjostrand18} and in the chaotic one~\cite{Anantharaman10}. Second, it is natural to focus on how eigenfrequencies can accumulate at the boundary of the strip and also 
to get resolvent estimates near the boundary of the strip. Again, this question has been explored both in the integrable case~\cite{AschLebeau, HitrikSjostrand04, BurqHitrik, AnantharamanLeautaud, BurqGerard18} and in the chaotic one~\cite{Christianson07, Schenck10, Nonnenmacher11, ChristiansonSchenckVasyWunsch, Riviere14, Jin17}.

The purpose of this work is to consider the second question for simple models of completely integrable systems. Via 
these models, we aim at illustrating the influence of the subprincipal symbol of the selfajoint part of our semiclassical operators on the asymptotic distribution of eigenvalues but also on resolvent estimates near the real axis. As briefly reminded below, this is related to the decay of the corresponding semigroup~\cite{Lebeau96}. Among other things, our study is motivated by earlier results due to Asch and Lebeau~\cite[Th.~2.3]{AschLebeau}. In that reference, they indeed showed how a selfadjoint pertubation of the principal symbol of the damped wave operator on the $2$-sphere can create a spectral gap inside the spectrum in the high frequency limit. Theorem~\ref{t:main_result_3} below shows how this result can be extended to our context\footnote{Observe that, compared with~\cite{AschLebeau}, our operators are not necessarily associated with a periodic flow.}. A major ingredient in the proof of~\cite{AschLebeau} but also in the works of Hitrik-Sj\"ostrand~\cite{HitrikSjostrand04, HitrikSjostrand05, HitrikSjostrandVuNgoc07, HitrikSjostrand08, HitrikSjostrand08b, HitrikSjostrand12, HitrikSjostrand18} is the \emph{analyticity} of the involved operators. One of the novelty of the present article compared with these references is Theorem~\ref{t:main_result_1} where we only suppose that the operators are \emph{smooth}, i.e. quantizing $\mathcal{C}^{\infty}$ symbols. This Theorem shows what can be said under these lower regularity assumptions and how this is influenced by the subprincipal symbols of the selfadjoint part as it was the case in~\cite{AschLebeau}. This will be achieved by building on the dynamical construction used by the first author and Maci\`a  for studying Wigner measures of semiclassical harmonic oscillators in~\cite{Ar_Mac18} -- see also~\cite{MaciaRiviere16, MaciaRiviere17} in the case of Zoll manifolds. As in~\cite{Ar_Mac18}, we restrict ourselves to the case of nonselfadjoint perturbations of semiclassical harmonic oscillators on $\mathbb{R}^d$. Yet it is most likely that the methods presented here can be adapted to deal with semiclassical operators associated with more general completely integrable systems, including damped wave equations on Zoll manifolds.

\subsection{Nonselfadjoint harmonic oscillators}

Let us now describe the spectral framework we are interested in. We fix $\omega=(\omega_1,\ldots,\omega_d)$ 
to be an element of $(\mathbb{R}_+^*)^d$ and we set $\widehat{H}_\hbar$ to be the semiclassical harmonic oscillator given by
\begin{equation}
\label{quantum_harmonic_oscillator}
\widehat{H}_\hbar: = \frac{1}{2} \sum_{j=1}^d \omega_j \big( - \hbar^2 \partial_{x_j}^2 +  x_j^2 \big). 
\end{equation}
We want to understand the spectral properties of nonselfadjoint perturbations of $\widehat{H}_\hbar$. Before being more precise on that issue, 
let us recall that the symbol $H$ of $\widehat{H}_\hbar$ is given by the classical harmonic oscillator:
\begin{equation}
\label{classical_harmonic_oscillator}
H(x,\xi) = \frac{1}{2}  \sum_{j=1}^d  \omega_j \big( \xi_j^2 +  x_j^2 \big), \quad (x,\xi) \in \R^{2d},
\end{equation}
whose induced Hamiltonian flow will be denoted by $\phi^H_t$. A brief account on the dynamical properties of this flow is given 
in paragraph~\ref{s:the_harmonic_ocillator}. For any smooth function $a \in \mathcal{C}^\infty(\R^{2d})$, we define its average $\langle a \rangle$ 
by the Hamiltonian flow $\phi_t^H$ as
\begin{equation}
\label{average}
\langle a \rangle(x,\xi) := \lim_{T \to \infty} \frac{1}{T} \int_0^T a \circ \phi_t^H(x,\xi) \, dt\ \in\mathcal{C}^{\infty}(\mathbb{R}^{2d}),
\end{equation}
whose properties are related to the Diophantine properties of $\omega$ -- see paragraph~\ref{s:the_harmonic_ocillator} for details.

Fix now two smooth functions $A$ and $V$ in $\mathcal{C}^{\infty}(\mathbb{R}^{2d},\mathbb{R})$ all of whose derivatives (at any order) are bounded. 
Following~\cite[Ch.~4]{Zw12}, one can define the Weyl quantization of these smooth symbols:
$$\widehat{A}_{\hbar}:=\Op_\hbar^w(A),\quad\text{and}\quad\widehat{V}_{\hbar}:=\Op_\hbar^w(V).$$
These are selfadjoint operators which are bounded on $L^2(\mathbb{R}^d)$ thanks to the Calder\'on-Vaillancourt Theorem. We 
aim at describing the asymptotic properties of the following nonselfadjoint operators in the semiclassical limit $\hbar\rightarrow 0^+$:
$$
\widehat{P}_{\hbar}:= \widehat{H}_\hbar +\delta_\hbar \widehat{V}_{\hbar}+ i \hbar  \widehat{A}_\hbar,
$$
where $\delta_\hbar\rightarrow 0$ as $\hbar\rightarrow 0^+$. More precisely, we focus on sequences of 
(pseudo-)eigenvalues $\lambda_\hbar=\alpha_\hbar + i  \hbar\beta_\hbar$ 
such that there exist $\beta\in\mathbb{R}$ and $(v_\hbar)_{\hbar\rightarrow 0^+}$ in $L^2(\mathbb{R}^d)$ for which
\begin{equation}\label{e:eigenvalue-equation}
(\alpha_\hbar, \beta_\hbar) \to (1, \beta), \quad \text{as } \hbar \to 0^+,\quad\text{and}
\quad\widehat{P}_{\hbar} \, v_\hbar = \lambda_\hbar \, v_\hbar+r_{\hbar},\quad \|v_{\hbar}\|_{L^2}=1.\end{equation}
Here $r_{\hbar}$ should be understood as a small remainder term which will be typically of order $o(\hbar)$. This remainder term allows us to encompass 
the case of quasimodes which is important to get resolvent estimates. 

\begin{remark}
 All along this work, we shall consider subsequences $\hbar_n \to 0^+$ so that the above convergence property holds. In order to alleviate notations, we will omit the index $n$ and just write $\hbar\to 0^+$, $\lambda_{\hbar}=\lambda_{\hbar_n}$, $v_\hbar=v_{\hbar_n}$, etc. For a similar reason, we do not relabel subsequences. This kind of conventions is standard when working with semiclassical parameters.
\end{remark}

Recall from the works of Markus-Matsaev~\cite{Markus88,Markus_Matsaev80} and Sj\"ostrand~\cite[Th.~5.2]{Sjostrand00} that true eigenvalues exist and that, counted with their algebraic multiplicity, they verify Weyl 
asymptotics as $\hbar\rightarrow 0^+$. It also follows from the works of Rauch-Taylor~\cite{RauchTaylor}, Lebeau~\cite{Lebeau96} and Sj\"ostrand~\cite[Lemma~2.1]{Sjostrand00} 
that
\begin{prop}\label{p:known_result} Let $(\lambda_\hbar=\alpha_\hbar +  i \hbar\beta_\hbar)_{\hbar\rightarrow 0^+}$ 
be a sequence verifying~\eqref{e:eigenvalue-equation} with $\beta_{\hbar}\rightarrow\beta$ and $r_{\hbar}=o(\hbar)$. Then, 
one has
\begin{equation}
\beta \in \left[ \min_{z \in H^{-1}(1)} \langle A \rangle(z),\max_{z \in H^{-1}(1)} \langle A \rangle(z) \right].
\end{equation}
\end{prop}
Note that one always has
$$\min_{z \in H^{-1}(1)} A(z)\leq A_-:=\min_{z \in H^{-1}(1)} \langle A \rangle(z)\leq A_+:=\max_{z \in H^{-1}(1)} \langle A \rangle(z)\leq \max_{z \in H^{-1}(1)}  A(z),$$
where the inequalities may be strict. For the sake of completeness and as it will be instructive for our proof, we briefly recall the proof of this proposition\footnote{In the case where the nonselfadjoint perturbation is $\gg\hbar$ and where the 
symbols enjoy some extra analytical properties, this proposition remains true (after a proper renormalization) 
when $r_{\hbar}=0$ and when $\omega$ satisfies appropriate diophantine properties as~\eqref{diophantine} below.} 
in paragraph~\ref{ss:proof-sjostrand}. One can verify that the quantum propagator $\displaystyle\left(e^{\frac{it\widehat{P}_{\hbar}}{\hbar}}\right)_{t\geq 0}$ defines a bounded operator on $L^2(\mathbb{R}^d)$ whose norm is bounded by $e^{|t|\|\Op_\hbar(A)\|_{\mathcal{L}(L^2)}}$. Moreover, if we suppose in addition that $\langle A\rangle \geq a_0>0$ on $\mathbb{R}^{2d}$, we say that the damping term is geometrically controlled and one gets exponential decay of the quantum propagator in time~\cite{Lebeau96, HelfferSjostrand10}. More generally, controlling the way pseudo-eigenvalues accumulate on the real axis provides informations on the decay rate of the quantum propagator~\cite{Lebeau96, HelfferSjostrand10}, and this is precisely the question we are aiming at when $\langle A\rangle$ may vanish.

\subsection{The smooth case}
Let us now explain our main results which show how the selfadjoint term $\widehat{V}_{\hbar}$ influences the way that the eigenvalues may 
accumulate on the boundary of the interval given by Proposition~\ref{p:known_result}. In the smooth case, our main result reads as follows:

\begin{teor}
\label{t:main_result_1}
Suppose that $A\geq 0$ and that, for every $(x,\xi)\in H^{-1}(1)\cap\langle A\rangle^{-1}(0)$, there exists $T > 0$ such that
\begin{equation}\label{e:control-subprincipal}\langle A\rangle\circ \phi^{\langle V\rangle}_{T}(x,\xi)>0,\end{equation}
where $\phi_t^{\langle V\rangle}$ is the Hamiltonian flow generated by $\langle V\rangle$. For every $R>0$, 
there exists\footnote{The (more or less explicit) constant $\varepsilon_R$ coming out from our proof verifies $\lim_{R\rightarrow+\infty}\varepsilon_R=0.$} $\varepsilon_R>0$ such that, for
$$\delta_{\hbar}\geq\varepsilon_R^{-1}\hbar^2,$$
and, for every sequence $(\lambda_\hbar=\alpha_\hbar +  i \hbar\beta_\hbar)_{\hbar\rightarrow 0^+}$ 
verifying~\eqref{e:eigenvalue-equation} with $\|r_{\hbar}\|\leq\varepsilon_R\hbar\delta_{\hbar}$,
$$\liminf_{\hbar\rightarrow 0^+}\frac{\beta_{\hbar}}{\delta_{\hbar}}> R.$$
\end{teor}

\begin{remark} If $\delta_{\hbar}\gg\hbar^2$ and $\|r_{\hbar}\|\ll\hbar\delta_{\hbar}$, then this Theorem shows 
that
 $$\lim_{\hbar\rightarrow 0^+}\frac{\beta_{\hbar}}{\delta_{\hbar}}=+\infty.$$
\end{remark}

In other words, under the geometric control condition~\eqref{e:control-subprincipal}, eigenvalues cannot accumulate too fast on the real axis as $\hbar\rightarrow 0^+$. 
We emphasize that, compared with the analytic case treated in~\cite{AschLebeau}, our result applies a priori to quasimodes. 
Hence, it also yields the following resolvent estimate in the smooth case. For every $R>0$, there exists 
some constant $\varepsilon_R>0$ such that, for $\hbar>0$ small enough and for $\delta_{\hbar}\geq\varepsilon_R^{-1}\hbar^2,$
\begin{equation}\label{e:resolvent}\frac{\text{Im}\ \lambda}{\hbar}\leq R\delta_{\hbar}\implies
\left\|\left(\widehat{P}_{\hbar}-\lambda\right)^{-1}\right\|_{L^2\rightarrow L^2}\leq\frac{1}{\varepsilon_R\hbar\delta_{\hbar}},\end{equation}
which is useful regarding energy decay estimates and asymptotic expansion of the corresponding 
semigroup -- see e.g.~\cite{HelfferSjostrand10}.

Note that 
the assumption that $A\geq 0$ makes the proof a little bit simpler but we could deal with more general functions by using the (nonselfadjoint) averaging method 
from~\cite{Sjostrand00} and by making some appropriate Diophantine assumptions -- see e.g. paragraph~\ref{averaging_method}. Our proof will crucially 
use the Fefferman-Phong inequality (hence the Weyl quantization) and this allows us to reach perturbations of size $\delta_{\hbar}\gtrsim \hbar^2$. If we had used another choice (say for instance the standard one), we would have only been able to use the Garding inequality and it would have lead us to the stronger restriction $\delta_{\hbar}\gtrsim \hbar$.

In the case where $V=0$ 
and under some analyticity assumptions in dimension $2$, it was shown by Hitrik and 
Sj\"ostrand~\cite[Th.~6.7]{HitrikSjostrand04} that one can find some eigenvalues such that $\beta_{\hbar}$ is exactly 
of order $\hbar$ provided that $\phi_t^H$ is periodic and that $\langle A\rangle$ vanishes on finitely many closed orbits. 
Hence, our hypothesis~\eqref{e:control-subprincipal} on the 
subprincipal $V$ is crucial here. Note that this geometric condition is similar to the one appearing in~\cite{Ar_Mac18} for the study of semiclassical 
measures of the Schr\"odinger equation -- see also~\cite{MaciaRiviere16, MaciaRiviere17} in the case of Zoll manifolds. 
As we shall see, ensuring this dynamical property depends on the Diophantine properties of $\omega$. Recall that, 
to each $\omega$, one can associate the submodule
\begin{equation}
\label{submodule_intro}
\Lambda_\omega := \{ k \in \mathbb{Z}^d \, : \,  \omega \cdot k = 0 \}.
\end{equation}
When the resonance module $\Lambda_\omega = \{ 0 \}$, we will see in paragraph~\ref{s:the_harmonic_ocillator} that our geometric control condition~\eqref{e:control-subprincipal} can only be satisfied if $\langle A\rangle >0$. 
A typical case in which our dynamical condition holds is when $H^{-1}(1)\cap\langle A\rangle^{-1}(0)$ consists in a disjoint union of a finite number of minimal $\phi_t^H$-invariant 
tori $(\mathcal{T}_k)_{k=1,\ldots N}$. In this case, our dynamical condition is equivalent to say that the Hamiltonian vector field $X_{\langle V \rangle}$ satisfies
$$
\forall 1\leq k\leq N,\quad\forall z\in\mathcal{T}_k,\quad X_{\langle V\rangle}(z)=\frac{d}{dt}\left(\phi_t^{\langle V\rangle}(z)\right)_{|t=0} \notin T_z\mathcal{T}_k.
$$

\subsection{The analytic case}

We now discuss the case where the functions $A$ and $V$ enjoy some analyticity properties. To that aim, we follow a method introduced by 
Asch and Lebeau in the case of the damped wave equation on the $2$-sphere~\cite{AschLebeau}. 
We will explain how to adapt this strategy in the framework of harmonic oscillators which are not necessarly periodic.  
The upcoming results should be viewed as an extension of Asch-Lebeau's cons\-truction to semiclassical harmonic oscillators and as an 
illustration on what can be gained via analyticity compared with the purely dynamical approach used to prove 
Theorem~\ref{t:main_result_1}. We emphasize that the argument presented here only holds for true 
eigenmodes, i.e. $r_\hbar=0$ in~\eqref{e:eigenvalue-equation}. In particular, it does not seem to yield any resolvent estimate like~\eqref{e:resolvent} 
which is crucial to deduce some results on the semigroup generated by $\widehat{P}_{\hbar}$.

We now assume some extra conditions on the symbols $H$, $V$ and $A$. First, given the vector 
of frequencies $\omega := (\omega_1, \ldots , \omega_d)$ 
of the harmonic oscillator $H$, we shall say that $\omega\in \R^d$ is \textit{partially Diophantine}~\cite[Eq.~(2.19)]{Llav03} if one has:
\begin{equation}
\label{diophantine}
\vert \omega \cdot k \vert^{-1} \leq C \vert k \vert^\nu, \quad \forall k \in \mathbb{Z}^d \setminus \Lambda_\omega.
\end{equation}
This restriction is due to the fact that, in the process of averaging, we will deal with the classical problem of small denominators in KAM theory. 
To keep an example in mind, note that $\omega=(1,\ldots,1)$ is obviously partially Diophantine\footnote{In that example, the flow is periodic and 
we are in the same situation as in~\cite{AschLebeau}.}.

We will make use of some analyticity assumptions on the symbols $V$ and $A$ in the following sense:
\begin{definition}
\label{spaces_analytic_functions}
Let $s > 0$. We say that  $a \in L^1(\R^{2d})$ belongs to the space $\mathcal{A}_s$ if
$$
\Vert a \Vert_{s} := \int_{\R^{2d}} \vert \widehat{a}(w) \vert e^{s \| w \|} \, dw < \infty,
$$
where $\widehat{a}$ denotes the Fourier transform of $a$ and $\| w\|$ the Euclidean norm on $\R^{2d}$.

Let $\rho,s> 0$, we introduce the space $\mathcal{A}_{\rho,s}$ of functions $a \in L^1(\R^{2d})$ such that
\begin{equation}
\label{weighted_norm}
\Vert a \Vert_{\rho ,s} := \frac{1}{(2\pi)^d} \sum_{k \in \mathbb{Z}^d} \Vert a_k \Vert_s \, e^{\rho \vert k \vert} < \infty,
\end{equation}
where
$$
a_k(z) = \int_{\mathbb{T}^d} a \circ \Phi_\tau^H(z) e^{-i k \cdot \tau} d \tau, \quad k \in \mathbb{Z}^d,
$$
with $\Phi_\tau^H$ defined by \eqref{e:multiflow}.
\end{definition}
\begin{remark} Observe that, for any $a$ element in $\mathcal{A}_s$ and for every multi-index $\alpha\in\mathbb{Z}_+^d$, $\widehat{\partial^\alpha a}$ belongs to $L^1$. Hence, $a$ is smooth and one has $\partial^\alpha a\in L^\infty$ for every $\alpha\in\mathbb{Z}_+^d$. Hence, any element in $\mathcal{A}_s$ belongs to the class $S(1)$ of symbols that are amenable to semiclassical calculus on $\mathbb{R}^d$. In particular, by~\cite[Lemma~4.10]{Zw12}, one has
\begin{equation}
\label{e:calderon_vaillancourt_for_analytic}
\forall\ a\in\mathcal{A}_s,\quad\Vert \Op_\hbar^w(a) \Vert_{\mathcal{L}(L^2)} \leq C_{d,s} \Vert a \Vert_s.
\end{equation}
As a consequence of~\eqref{e:harmonic_fourier-decomposition},
one can show that:
$\Vert a \Vert_s \leq \Vert a \Vert_{\rho,s}, \quad \forall \rho > 0.$
\end{remark}

Our next result reads:
\begin{teor}
\label{t:main_result_3} Suppose that $A$ and $V$ belong to the space $\mathcal{A}_{\rho,s}$ 
for some fixed $\rho, s > 0$ and that $\langle A\rangle \geq 0$. Assume also that $\omega$ is partially Diophantine and that, 
for every $(x,\xi) \in H^{-1}(1)\cap\langle A\rangle^{-1}(0)$, there exists $T > 0$ such that
$$\langle A\rangle\circ \phi^{\langle V\rangle}_{T}(x,\xi)>0.$$
Then there exists $\varepsilon: = \varepsilon( A ,  V  ) > 0$ such that, for
$$\delta_\hbar = \hbar,$$
and for any sequence of solutions to~\eqref{e:eigenvalue-equation} with 
$r_{\hbar}=0$,
\begin{equation}
\beta \geq \varepsilon.
\end{equation}
\end{teor}

This Theorem shows that eigenvalues of the nonselfadjoint operator $\widehat{P}_{\hbar}$ cannot accumulate on the boundary of the strip given by 
Proposition~\ref{p:known_result}. Compared with Theorem~\ref{t:main_result_1}, 
it only deals with the case of true eigenvalues and it does not seem that a good resolvent 
estimate can be easily deduced from the proof below. Finally, for the sake of simplicity, we also supposed that $\delta_{\hbar}=\hbar$ but it is most likely 
that the argument can be applied when $\delta_{\hbar}$ does not go to $0$ too slowly.

\subsection*{Acknowledgements} We are grateful to S.~Nonnenmacher for suggesting to consider how the methods from~\cite{MaciaRiviere16, MaciaRiviere17} can be applied in the nonselfadjoint context. We warmly thank J.~Sj\"ostrand and the anonymous referee for their useful comments on the article. Part of this work was achieved when V. Arnaiz was visiting 
the Laboratoire Paul Painlev\'e in Lille. V. Arnaiz has been supported by La Caixa, Severo Ochoa ICMAT, International Phd. Programme, 2014, and MTM2017-85934-C3-3-P (MINECO, Spain). G. Rivi\`ere has been partially supported by the ANR project 
GERASIC (ANR-13-BS01-0007-01) and the Labex CEMPI (ANR-11-LABX-0007-01).

\section{The classical Harmonic Oscillator}
\label{s:the_harmonic_ocillator}

The Hamiltonian equations corresponding to $H$ are given by
\begin{equation}
\label{hamilton-system}
\left \lbrace \begin{array}{l}
\dot{x}_j = \omega_j \xi_j, \\[0.2cm]
\dot{\xi}_j = - \omega_j x_j, \quad j=1, \ldots , d.
\end{array} \right.
\end{equation}
Hence, we can write the solution to this system as a superposition of $d$-independent commuting flows as follows:
$$
\big( x(t), \xi(t) \big) = \phi_t^H(x,\xi) := \phi_{\omega_d t}^{H_d} \circ \cdots \circ \phi_{\omega_1 t}^{H_1}(x,\xi), \quad (x,\xi) \in \R^{2d}, \quad t \in \R,
$$
where $H_j(x,\xi)=\frac{1}{2}(x_j^2+\xi_j^2)$ and where $\phi_t^{H_j}(x,\xi)$ denotes the associated Hamiltonian flow. 
In other words, the solution to \eqref{hamilton-system} can be written in terms of the unitary block matrices
\begin{equation}
\label{matrix}
\left( \begin{array}{c} x_j(t) \\[0.2cm]
\xi_j(t) \end{array} \right) = \left( \begin{array}{cc}
\cos( \omega_j t) &  \sin(\omega_j t) \\[0.2cm]
- \sin(\omega_j t) & \cos(\omega_j t) \end{array}
\right) \left( \begin{array}{c} x_j \\[0.2cm]
\xi_j \end{array} \right), \quad j=1, \ldots , d.
\end{equation}
Observe that each flow $\phi_t^{H_j}$ is periodic with period $2\pi$. We now introduce the transformation:
\begin{equation}
\label{e:multiflow}
\Phi_{\tau}^H := \phi_{t_d}^{H_d} \circ \cdots \circ \phi_{t_1}^{H_1}, \quad \tau=(t_1, \ldots, t_d) \in \R^d.
\end{equation}
Note that $\tau\mapsto\Phi_{\tau}^H$ is $2\pi\mathbb{Z}^d$-periodic; therefore we can view it as a function on the torus $\mathbb{T}^d:=\mathbb{R}^d/2\pi\mathbb{Z}^d$. Considering now 
the submodule $$\Lambda_\omega:=\left\{k\in\mathbb{Z}^d:k \cdot \omega=0\right\},$$ we can define the minimal torus
$$
\mathbb{T}_\omega := \Lambda_\omega^\perp/(2\pi \mathbb{Z}^d \cap \Lambda_\omega^\perp),
$$
where $\Lambda_\omega^\perp$ denotes the linear space orthogonal to $\Lambda_\omega$. The dimension of $\mathbb{T}_\omega$ is $d_\omega = d - \operatorname{rk} \Lambda_\omega$. Kronecker's theorem 
states that the family of probability measures on $\mathbb{T}^d$ defined by
$$\frac{1}{T}\int_0^T \delta_{t\omega} \,dt$$ 
converges (for the weak-$\star$ topology) to the normalized Haar measure $\nu_\omega$ on the subtorus $\mathbb{T}_\omega\subset \mathbb{T}^d$.

For any function $a \in \mathcal{C}^\infty(\R^{2d})$, $a\circ\phi_{t}^H=a\circ\Phi_{t\omega}^H$. Thus, we can write the average $\langle a \rangle$ of $a$ by the flow $\phi_t^H$ as
\begin{equation}
\label{average-formula}
\langle a \rangle(x,\xi) = \lim_{T\to\infty}\frac{1}{T}\int_0^T a\circ\Phi_{t\omega}^H(x,\xi) dt=\int_{\mathbb{T}_\omega}a\circ\Phi_\tau^H(x,\xi)\nu_\omega(d\tau)\in\mathcal{C}^\infty(\mathbb{R}^{2d}).
\end{equation}
Recall that the energy hypersurfaces $H^{-1}(E) \subset \R^{2d}$ are compact for every $E \geq 0$. 
For $E>0$, due to the complete integrability of $H$, these hypersurfaces are foliated by the 
invariant tori: $\{ \Phi_\tau^H(x,\xi):\tau\in\mathbb{T}_{\omega}\}$. Note that some invariant tori 
of the energy hypersurface $H^{-1}(E)$, $E > 0$, may have dimension less than $d_\omega$. 
For instance, if $\omega = (1,\pi)$ then $d_\omega = 2$, 
but the torus $\{ \Phi_\tau^H(0,1,0,1):\tau\in\mathbb{T}_{\omega}\} \subset H^{-1}(\pi)$ 
has dimension $1$. 
 
Observe also that $1\leq d_{\omega}\leq d$. In the case $d_{\omega}=1$ and $\omega=\omega_1(1,\ldots,1)$, the flow $\phi_t^H$ is $2\pi/\omega_1$-periodic. On the other hand, if $d_\omega = d$, then, for 
every $a\in\mathcal{C}^{\infty}(\mathbb{R}^{2d})$, there exists $\mathcal{I}(a)\in\mathcal{C}^\infty(\R^{d})$ such that
$\langle a \rangle(z) = \mathcal{I}(a)(H_1(z), \ldots, H_d(z)).$
In particular, for every $a$ and $b$ in $\mathcal{C}^{\infty}(\mathbb{R}^{2d})$, one has $\{\langle a \rangle,\langle b \rangle\}=0$ whenever $d_{\omega}=d$.

To conclude this section, we prove the following lemma:

\begin{lemma}
\label{average_analytic_norm}
If $a \in \mathcal{A}_{s}$ then $\langle a \rangle \in \mathcal{A}_{s}$ and $\Vert \langle a \rangle \Vert_s \leq \Vert a \Vert_{s}$.
\end{lemma}
\begin{proof}
By \eqref{average-formula}, we can write the Fourier transform of $\langle a \rangle$ as
$$
\widehat{ \langle a \rangle}(x,\xi)  = \int_{\mathbb{T}_\omega} \widehat{ a \circ \Phi_\tau^H}(x,\xi) \nu_\omega(d\tau).
$$
Moreover, since $\widehat{ a \circ \Phi_\tau^H}(x,\xi) = \widehat{a} \circ \Phi_{\tau}^H(x,\xi)$ thanks to~\eqref{matrix}, 
we have that $\widehat{ \langle a \rangle} =  \langle \widehat{a} \rangle$. Thus, using~\eqref{matrix} one more time, one finds
\begin{align*}
\Vert \langle a \rangle \Vert_s & = \int_{\R^{2d}} \vert \langle \widehat{a} \rangle (z) \vert e^{s \vert z \vert} dz \\
 & \leq  \int_{\mathbb{T}_\omega} \int_{\R^{2d}} \vert  \widehat{a} \circ \Phi_{\tau}^H(z) \vert   e^{s \vert z \vert} dz \, \nu_\omega(d\tau)\\
 & = \int_{\R^{2d}} \vert  \widehat{a} (z) \vert   e^{s \vert z \vert} dz = \Vert a \Vert_s.
\end{align*}
\end{proof}

\section{Proof of Theorem~\ref{t:main_result_1}}

We now give the proof of our main result in the $\mathcal{C}^{\infty}$ case. Before doing that, we briefly recall the proof of Proposition~\ref{p:known_result} in order to make the 
proof of Theorem~\ref{t:main_result_1} more comprehensive. Note that we use the following convention for the scalar product on $\R^d$:
$$\langle u,v\rangle_{L^2}=\int_{\R^d}u(x)\overline{v(x)}dx.$$

\subsection{Proof of Proposition~\ref{p:known_result}}\label{ss:proof-sjostrand} Let $\lambda_\hbar=\alpha_\hbar+i\hbar\beta_{\hbar}$ be a sequence of 
(pseudo-)eigenvalues verifying~\eqref{e:eigenvalue-equation}. Denote by 
$(v_{\hbar})_{\hbar\rightarrow 0^+}$ the corresponding sequence of normalized quasimodes. Introduce the Wigner distribution $W^\hbar_{v_\hbar} \in \mathcal{D}'(\R^{2d})$ associated to the function $v_\hbar$:
$$
W_{v_\hbar}^\hbar : \mathcal{C}_c^\infty(\R^{2d}) \ni a \longmapsto W^\hbar_{v_\hbar}(a) := \langle \Op_\hbar^w(a) v_\hbar , v_\hbar \rangle_{L^2(\R^d)}.
$$
According to~\cite[Ch.~5]{Zw12} and modulo extracting a subsequence, there exists a probability measure $\mu$ carried by $H^{-1}(1)$ such that $W^\hbar_{v_\hbar} \rightharpoonup \mu.$ 
The measure $\mu$ is called the semiclassical measure associated to the (sub)sequence $(v_\hbar)_{\hbar\rightarrow 0^+}$. Note that these properties of the limit points follow 
from the facts that $v_\hbar$ is normalized and that $\widehat{H}_{\hbar}v_{\hbar}=v_{\hbar}+o_{L^2}(1)$. We will now make use of the eigenvalue equation~\eqref{e:eigenvalue-equation} 
to derive an invariance property of $\mu$. Using the symbolic calculus for Weyl pseudodifferential operators~\cite[Ch.~4]{Zw12}, we have, for every 
$a \in \mathcal{C}_c^\infty(\R^{2d},\mathbb{R})$,
\begin{align*}
\big \langle [\widehat{H}_\hbar + \delta_\hbar  \widehat{V}_\hbar , \Op_\hbar^w(a)] v_\hbar, v_\hbar \big \rangle_{L^2(\R^{d})} & = 
\frac{\hbar}{i} \big \langle \Op_\hbar^w(\{ H , a \}) v_\hbar, v_\hbar \big \rangle_{L^2(\R^{d})} + O(\hbar(\delta_\hbar +\hbar)).
\end{align*}
On the other hand, using that $v_\hbar$ is a quasimode of $\widehat{P}_\hbar$ and the composition rule for the Weyl quantization~\cite[Ch.~4]{Zw12}, we also have
$$
\big \langle [\widehat{H}_\hbar + \delta_\hbar  \widehat{V}_\hbar , \Op_\hbar^w(a)] v_\hbar, v_\hbar \big \rangle_{L^2(\R^{d})} = 
2i \hbar \,  \big \langle  \Op_\hbar^w(a(A-\beta_{\hbar})) v_\hbar , v_\hbar \big \rangle_{L^2(\R^d)} +O(\|r_{\hbar}\|)+ O(\hbar^3).$$
Note that there is no $O(\hbar^2)$ term due to the fact that $a$ is real valued and to the symmetries of the Weyl quantization. 
Passing to the limit $\hbar\rightarrow 0^+$ and recalling that $\|r_{\hbar}\|=o(\hbar)$, one finds that $\mu(\{H,a\})=2\mu((\beta-A)a)$ for every $a$ in $\mathcal{C}^{\infty}_c(\R^d)$. This is equivalent to the fact that, for every $t\in\R$ 
and for every $a\in\mathcal{C}^{\infty}_c(\R^{2d})$, one has
\begin{equation}
\label{e:weak_solution}
\int_{\R^{2d}} a(z) \mu(dz) = \int_{\R^{2d}} a\circ \phi_{t}^H(z) \, e^{ 2 \int_0^t(A-\beta)\circ\phi_s^H(z)ds} \mu(dz). 
\end{equation}
Taking $a$ to be equal to $1$ in a neighborhood of $H^{-1}(1)$, identity \eqref{e:weak_solution} implies
\begin{equation}
\label{e:averaged_text_functions}
 e^{2\beta t} = \int_{\R^{2d}}  e^{ 2 \int_0^tA\circ\phi_s^H(z)ds}\mu(dz) , \quad \forall t \in \R ,
 \end{equation}
from which Proposition~\ref{p:known_result} follows thanks to~\eqref{average}. In the case, where $\beta=0$ and $A\geq 0$, 
one can deduce from~\eqref{e:averaged_text_functions} that
$$\forall t\in\R,\quad\text{supp}(\mu)\subset H^{-1}(1)\cap\{z:A\circ\phi_t^H(z)=0\}.$$
Hence, we can record the following useful lemma:
\begin{lemma}
\label{l:support_of_measure}
Suppose that $A\geq 0$. Let $\mu$ be a semiclassical measure associated to the sequence $(v_\hbar)_{\hbar\rightarrow 0^+}$ satisfying~\eqref{e:eigenvalue-equation} with $\beta=0$ and $r_{\hbar}=o(\hbar)$. Then
\begin{equation}
\label{e:support}
\supp \mu \subset \left \{ z \in H^{-1}(1) \, : \, \langle A \rangle(z) =0 \right \}.
\end{equation}
\end{lemma}

\subsection{Proof of Theorem~\ref{t:main_result_1}} Let us now reproduce the same argument but suppose now that $a=\langle a\rangle$, 
implying in particular that $\{H,\langle a\rangle\}=0$. From this, we get
\begin{align*}
\big \langle [\widehat{H}_\hbar + \delta_\hbar  \widehat{V}_\hbar , \Op_\hbar^w(\langle a\rangle)] v_\hbar, v_\hbar \big \rangle_{L^2(\R^{d})} & = 
\frac{\hbar\delta_{\hbar}}{i} \big \langle \Op_\hbar^w(\{ V , \langle a\rangle \}) v_\hbar, v_\hbar \big \rangle_{L^2(\R^{d})} + O(\hbar^3).
\end{align*}
As before, recalling that $a$ is real valued, one still has
$$
\big \langle [\widehat{H}_\hbar + \delta_\hbar  \widehat{V}_\hbar , \Op_\hbar^w(\langle a\rangle)] v_\hbar, v_\hbar \big \rangle_{L^2(\R^{d})}  =  
2i \hbar \,  \big \langle  \Op_\hbar^w(\langle a\rangle(A-\beta_{\hbar})) v_\hbar , v_\hbar \big \rangle_{L^2(\R^d)}+O(\|r_{\hbar}\|)+ O(\hbar^3).$$
Hence, one gets
$$ \big \langle  \Op_\hbar^w\left(\left(2(A-\beta_{\hbar})+\delta_{\hbar}X_V\right)\langle a\rangle\right) v_\hbar , v_\hbar \big \rangle_{L^2(\R^d)}
=O(\|r_{\hbar}\|\hbar^{-1})+O(\hbar^2),$$
where $X_V$ is the Hamiltonian vector field of $V$. Suppose now that $A\geq 0$ and $\langle a\rangle\geq 0$. From the Fefferman-Phong inequality~\cite[Ch.~4]{Zw12}, one knows that there exists some constant $C>0$ such that
$$2\beta_{\hbar}\big \langle  \Op_\hbar^w\left(\langle a\rangle\right) v_\hbar , v_\hbar \big \rangle_{L^2(\R^d)}\geq 
\delta_{\hbar}\big \langle  \Op_\hbar^w\left(X_V\langle a\rangle\right) v_\hbar , v_\hbar \big \rangle_{L^2(\R^d)}-
C(\hbar^2+\|r_{\hbar}\|\hbar^{-1}),$$
where the constant $C$ depends only on $A$, $V$ and $a$. 
Now, we fix $R>0$ and we would like to show that 
$\liminf_{\hbar\rightarrow 0^+}\beta_{\hbar}/\delta_{\hbar}> R$ provided that
$\delta_{\hbar}\geq \varepsilon_R^{-1}\hbar^2$ and that $\|r_{\hbar}\| \leq \varepsilon_R  \hbar\delta_{\hbar}$ 
for some small enough $ \varepsilon_R>0$ (to be determined later on).
To that end, we proceed by contradiction and 
suppose that, up to an extraction, one has $2\frac{\beta_{\hbar}}{\delta_{\hbar}}\rightarrow c_0\in[0,2R]$ 
(in particular $\beta=0$). One finally gets after letting 
$\hbar\rightarrow 0^+$:
\begin{equation}\label{e:control-inequality-0}c_0 \mu\left(\langle a\rangle\right)
\geq\mu\left(X_V\langle a\rangle\right) - C\varepsilon_R,\end{equation}
for some $C\geq 0$ depending on $A$, $V$ and $a$. 
Using one more time Lemma~\ref{l:support_of_measure}, one can also deduce that $\mu$ is invariant by $\phi_t^H$. Hence, 
$$\mu\left(\{V,\langle a\rangle \}\right) = \mu\left(\{\langle V\rangle ,\langle a\rangle \}\right),$$
which implies
\begin{equation}\label{e:control-inequality}c_0 \mu\left(\langle a\rangle\right)
\geq\mu\left(X_{\langle V\rangle}\langle a\rangle\right) - C\varepsilon_R.\end{equation}
By our geometric control condition~\eqref{e:control-subprincipal} and 
since $H^{-1}(1) \cap \langle A \rangle^{-1}(0)$ is compact, there exist $T_1>0$ and $\varepsilon_0>0$ such that
$$\int_0^{T_1}\langle A\rangle\circ \phi_t^{\langle V\rangle}(z)dt > \varepsilon_0, \quad \forall z \in H^{-1}(1) \cap \langle A \rangle^{-1}(0),$$
where $\phi_t^{\langle V \rangle}$ is the flow generated by $X_{\langle V \rangle}$. 
Up to the fact that we may have to increase the value of $C>0$ (in a way that depends only 
on $T_1$, $A$, $V$ and $a$), 
we can suppose that~\eqref{e:control-inequality} holds uniformly for every function 
$\langle a\rangle\circ\phi_t^{\langle V\rangle}$ with $0\leq t\leq T_1$, i.e. for every $t\in[0,T_1]$,
$$
c_0 \mu\left(\langle a\rangle\circ\phi_t^{\langle V \rangle}\right)\geq\mu\left(\{\langle V \rangle, \langle a\rangle\}\circ\phi_t^{\langle V \rangle}\right) - C\varepsilon_R.
$$
This is equivalent to the fact that 
$\frac{d}{dt}\left(e^{-c_0 t}\int_{\R^{2d}}\langle a\rangle \circ\phi_t^{\langle V \rangle} d\mu\right)\leq 
C\varepsilon_Re^{-c_0t}$ for every $t\in[0,T_1]$. 
Hence, if $c_0\neq 0$, one finds that, for every $t\in [0,T_1]$,
\begin{equation}
\label{e:quasicontradiction}
\int_{\R^{2d}} \langle a \rangle \circ \phi_t^{\langle V \rangle}(z) \mu(dz) 
\leq e^{c_0 t} \int_{\R^{2d}} \langle a \rangle (z) \mu(dz) + 
\frac{C\varepsilon_R(e^{tc_0} - 1)}{c_0}.
\end{equation}
We now apply this inequality with $a=A$ and integrate over the interval $[0,T_1]$. In that way, we obtain 
$$\varepsilon_0<\int_{0}^{T_1}\int_{\R^{2d}} \langle a \rangle \circ \phi_t^{\langle V \rangle}(z) \mu(dz)dt
\leq\int_0^{T_1}\frac{C\varepsilon_R(e^{tc_0} - 1)}{c_0}dt \leq     \frac{C\varepsilon_RT_1(e^{T_1c_0} - 1)}{c_0}.$$
Observe that, for $c_0=0$, we would get the upper bound $C\varepsilon_RT_1^2.$ In both cases, this yields the 
expected contradiction by taking $\varepsilon_R$ small enough (in a way that depends only on $R$, $A$ and $V$) and it concludes the proof of Theorem~\ref{t:main_result_1}.

\begin{remark} Note that we could get the conclusion faster under the stronger geometric assumption
\begin{equation}\label{e:strong-control-condition}\forall z\in H^{-1}(1)\cap\langle A\rangle^{-1}(0),\quad \{\langle A\rangle,\langle V\rangle\}(z)\neq 0,\end{equation}
 which implies (but is not equivalent to) the geometric control condition~\eqref{e:control-subprincipal} of 
 Theorem~\ref{t:main_result_1}. Together with~\eqref{e:control-inequality}, this yields the following upper bound
 $$\mu\left(X_{\langle V\rangle}\langle A\rangle\right) \leq C\varepsilon_R.$$
 Hence, provided $\varepsilon_R>0$ is chosen small enough in a way that depends only on 
 $A$ and $V$ (but not $R$), we get a contradiction. This shows that, for a small enough choice 
 of $\varepsilon_R>0$, one has in fact $\beta_{\hbar}\gg\delta_{\hbar}$ under the geometric condition~\eqref{e:strong-control-condition}.
\end{remark}

\section{The Averaging Method}
\label{averaging_method}

From this point on of the article, we will make the assumption that
$$\delta_{\hbar}=\hbar.$$
This will slightly simplify the exposition and it should a priori be possible to extend the results provided $\hbar\leq\delta_{\hbar}$ does not go to $0$ too slowly. 
In this paragraph, we briefly recall how to perform a semiclassical averaging method in the context of nonselfadjoint operators 
following the works of Sj\"ostrand~\cite{Sjostrand00} and Hitrik~\cite{Hitrik02}. For that purpose, we define
$$\widehat{F}_{\hbar}:=\Op_\hbar^w( F_1+i F_2),$$
where $F_1$ and $F_2$ are two real valued and smooth functions on $\R^{2d}$ that will be determined later on. We make the assumption 
that all the derivatives (at every order) of $F_1$ and $F_2$ are bounded. For every $t$ in $[0,1]$, 
we set $\mathcal{F}_{\hbar}(t)=e^{it\widehat{F}_{\hbar}}$.
By \cite[Thm. III.1.3.]{Eng00}, the family $\mathcal{F}_\hbar(t)$ defines a strongly continuous group (note that $\mathcal{F}_\hbar$ is invertible) on $L^2(\R^d)$ such that
\begin{equation}
\label{e:bound_semigroup}
\Vert \mathcal{F}_\hbar(t) \Vert_{\mathcal{L}(L^2)} \leq e^{|t| \Vert\Op_\hbar^w(F_2) \Vert_{\mathcal{L}(L^2)}}.
\end{equation}
For simplicity, we shall denote $\mathcal{F}_\hbar=\mathcal{F}_\hbar(1)$ and we will study the properties of the conjugated operator
$$\widehat{Q}_{\hbar}:=\mathcal{F}_\hbar\widehat{P}_{\hbar}\mathcal{F}_\hbar^{-1},$$
for appropriate choices of $F_1$ and $F_2$. Using the conventions of~\cite[Ch.~4]{Zw12}, symbols of order $m\in\R$ are defined by
$$S\left(\langle z\rangle^m\right):=\left\{(a_{\hbar})_{0\leq\hbar\leq 1}\in\mathcal{C}^{\infty}(\R^{2d},\mathbb{C}):\ \forall\alpha\in\mathbb{N}^{2d},\ |\partial^{\alpha}a(z)|\leq 
C_{\alpha}\langle z\rangle^m\right\},$$
where $\langle z\rangle=(1+\|z\|^2)^{\frac{1}{2}}$. We shall denote by $\Psi_{\hbar}^m$ the set of all operators of the form $\Op_{\hbar}^w(a)$ with $a\in 
S\left(\langle z\rangle^m\right)$. 

\subsection{Semiclassical conjugation}

Writing the Taylor expansion, one knows that, 
for every $a$ in $S\left(\langle z\rangle^m\right)$,
\begin{equation}
\label{e:taylor_expansion}
\begin{array}{rl}
\mathcal{F}_\hbar\Op_{\hbar}^w(a)\mathcal{F}_\hbar^{-1} = & \hspace*{-0.2cm} \Op_{\hbar}^w(a)+i\left[\widehat{F}_{\hbar},\Op_{\hbar}^w(a)\right] \\[0.2cm]
  & \hspace*{-0.2cm} \displaystyle -
\int_0^1(1-t)\mathcal{F}_\hbar(t)\left[\widehat{F}_{\hbar},\left[\widehat{F}_{\hbar},\Op_{\hbar}^w(a)\right]\right]\mathcal{F}_\hbar(-t)dt.
\end{array}
\end{equation}
Observe from the composition rules for semiclassical pseudodifferential operators~\cite[Ch.~4]{Zw12} that 
$\left[\widehat{F}_{\hbar},\left[\widehat{F}_{\hbar},\Op_{\hbar}^w(a)\right]\right]$ is an element of $\hbar^2\Psi_{\hbar}^m$. Moreover, a direct extension of the 
Egorov Theorem~\cite[Th.~11.1]{Zw12} to the nonselfadjoint framework shows that the third term in the righthand side is in fact an 
element of $\hbar^2\Psi_{\hbar}^m$. Then one can verify from the composition rules for pseudodifferential operators that
$$\mathcal{F}_\hbar\Op_{\hbar}^w(a)\mathcal{F}_\hbar^{-1}=\Op_{\hbar}^w(a)+\hbar\Op_{\hbar}^w\left(\{F_1, a\}\right)+i\hbar
\Op_{\hbar}^w\left(\{F_2, a\}\right)+\hbar^2 \widehat{R}_\hbar,$$
where $\widehat{R}_\hbar$ is an element in $\Psi_{\hbar}^m$. Applying this equality to the 
operator $\widehat{P}_{\hbar}$, one finds
\begin{equation}\label{e:conjugation}
\widehat{Q}_{\hbar}=\widehat{P}_{\hbar}+\hbar\Op_{\hbar}^w\left(\{F_1, H\}\right)+i\hbar
\Op_{\hbar}^w\left(\{F_2, H\}\right)+\hbar^2 \widehat{R}_\hbar,
\end{equation}
where $\widehat{R}_\hbar$ is now an element in $\Psi_{\hbar}^2$. We now aim at 
choosing $F_1$ and $F_2$ in such a way that
\begin{equation}\label{e:cohomological}
 \{F_1,H\}+V=\langle V\rangle\quad\text{and}\quad\{F_2,H\}+A=\langle A\rangle.
\end{equation}
If we are able to do so, then we will have
\begin{equation}\label{e:conjugation2}
\mathcal{F}_\hbar\widehat{P}_{\hbar}\mathcal{F}_\hbar^{-1}=\widehat{H}_{\hbar}+\hbar\Op_{\hbar}^w\left(\langle V\rangle\right)+i\hbar
\Op_{\hbar}^w\left(\langle A\rangle\right)+\hbar^2 \widehat{R}_\hbar.
\end{equation}

\subsection{Solving cohomological equations}\label{cohomological-equations}

In order to solve cohomological-type equations like~\eqref{e:cohomological}, we need to make a few Diophantine restrictions on $\omega$.
Let $g \in \mathcal{C}^\infty(\R^{2d})$ be any smooth function such that $\langle g \rangle = 0$ and all of whose derivatives (at any order) are bounded. We look for another function 
$f \in \mathcal{C}^\infty(\R^{2d})$ all of whose derivatives (at any order) are bounded and which solves the following cohomological equation:
\begin{equation}
\label{cohomological}
\{ H , f \} = g.
\end{equation}
We then apply this result with $g=V-\langle V\rangle$ (resp. $A-\langle A\rangle$) in order to find $f=F_1$ (resp. $F_2$).

For any $f \in \mathcal{C}^\infty(\R^{2d})$ all of whose derivatives (at any order) are bounded, we can write $f\circ\Phi_{\tau}^H$ as a Fourier series in $\tau \in \mathbb{T}^d$:
\begin{equation}
\label{e:harmonic_fourier-decomposition}
f\circ\Phi_{\tau}^H(x,\xi)=\sum_{k \in \mathbb{Z}^d}f_k(x,\xi) \frac{e^{ik\cdot\tau}}{(2\pi)^d}  ,\quad  f_k(x,\xi):=\int_{\mathbb{T}^d} f \circ \Phi_{\tau}^H(x, \xi) e^{-ik \cdot \tau} d\tau.
\end{equation}
Notice that $f_k \circ \Phi_\tau^H = f_k \, e^{ik \cdot \tau}$ and that, for $\tau = 0$, $f = (2\pi)^{-d}\sum_k f_k$. Recalling~\eqref{average-formula} and the definition~\eqref{submodule_intro} of $\Lambda_{\omega}$, one has
$$
\langle f \rangle\circ\Phi_{\tau}^H(x,\xi)  =\sum_{k\in\mathbb{Z}^d}f_k(x,\xi)\left(\lim_{T\rightarrow+\infty}\frac{1}{T(2\pi)^d}\int_0^T e^{ik.(\tau+t\omega)}dt\right) = \frac{1}{(2\pi)^d}\sum_{k \in \Lambda_\omega}f_k(x,\xi) e^{ik \cdot\tau} .
$$
In particular, as $\langle g\rangle\circ\Phi_\tau^H=0$ for every $\tau\in\mathbb{T}^d$, one finds that $g_k=0$ for every $k\in\Lambda_\omega$ and thus
$$\forall \tau\in\mathbb{T}^d,\quad g\circ\Phi_{\tau}^H(x,\xi)=\frac{1}{(2\pi)^{d}} \sum_{k \in \mathbb{Z}^d \setminus \Lambda_\omega} g_k(x,\xi)e^{ik \cdot\tau}.$$ 
Observe also that, if $f$ is a solution of~\eqref{cohomological}, then so is $f + \lambda \langle f \rangle$ for any $\lambda \in \mathbb{R}$, since $\{ H, \langle f \rangle \} = 0$ thanks to~\eqref{average-formula}. 
Thus, we can try to solve the cohomological equation~\eqref{cohomological} by supposing $f\circ\Phi_{\tau}^H$ to be of the form
$$
f\circ\Phi_{\tau}^H(x,\xi) = \frac{1}{(2\pi)^{d}}\sum_{k \in \mathbb{Z}^d \setminus \Lambda_\omega} f_k(x,\xi)e^{ik \cdot\tau},
$$
and write down
$$
\{ H, f\circ\Phi_{\tau}^H \} = \frac{d}{dt} \left(f \circ \Phi_{\tau+t\omega}^H \right) \vert_{t=0} = \frac{1}{(2\pi)^{d}} \sum_{k \in \mathbb{Z}^d \setminus \Lambda_\omega} ik \cdot \omega \, f_ke^{ik \cdot\tau}.
$$
Hence, if we set
\begin{equation}
\label{solution-cohomological}
f\circ\Phi_{\tau}^H(x,\xi) = \frac{1}{(2\pi)^d} \sum_{k \in \mathbb{Z}^d \setminus \Lambda_\omega} \frac{1}{ik\cdot \omega} \, g_k(x,\xi)e^{ik \cdot\tau},
\end{equation}
then $f$ will solve~\eqref{cohomological} (at least formally).
It is not difficult to see that, unless we impose some quantitive restriction on how fast $\vert k \cdot \omega \vert^{-1}$ can grow, the solutions given formally by (\ref{solution-cohomological}) 
may fail to be even distributions -- see for instance \cite[Ex.~2.16]{Llav03}. On the other hand, 
if $\omega$ is partially Diophantine, 
and $g \in \mathcal{C}^\infty(\R^{2d})$ has all its derivatives (at any order) bounded and is such that $\langle g \rangle = 0$, then \eqref{solution-cohomological} defines a smooth solution $f \in \mathcal{C}^\infty(\R^{2d})$ 
of~\eqref{cohomological} all of whose derivatives (at any order) are bounded. As a special case, we observe that, if $\omega=(1,\ldots,1)$, then an explicit solution of~\eqref{cohomological} is given by
\begin{equation}\label{lemma-periodic-case}
f = \frac{-1}{2\pi} \int_0^{2\pi} \int_0^t  g \circ \phi_s^H  \, ds \, dt.
\end{equation}

\subsection{Proof of Theorem~\ref{t:main_result_3}}

We now turn to the proof of Theorem~\ref{t:main_result_3} and to that aim, we should exploit 
the analyticity assumptions on $A$ and $V$ in order to improve the result of Theorem \ref{t:main_result_1} when $r_{\hbar}=0$ 
in~\eqref{e:eigenvalue-equation}. It means that we are not considering anymore quasimodes but true eigenmodes. 
Hence, from this point on of the article, 
$$r_{\hbar}=0.$$
The point of using analyticity is that the symbolic calculus on the family of spaces $\mathcal{A}_s$ is extremely well 
behaved -- see appendix~\ref{a:analytic} for a brief review. This will allow us to construct a second normal form for the operator $\widehat{P}_\hbar$ via conjugation by a second operator so 
that the nonselfadjoint part of the operator is averaged by the two flows $\phi_t^H$ and $\phi_t^{\langle V \rangle}$.

Recall from~\eqref{e:conjugation2} that
\begin{equation}\label{e:averaged-op-analytic}\mathcal{F}_\hbar\widehat{P}_{\hbar}\mathcal{F}_\hbar^{-1}=\widehat{H}_{\hbar}+\hbar\Op_{\hbar}^w\left(\langle V\rangle\right)+i\hbar
\Op_{\hbar}^w\left(\langle A\rangle\right)+\hbar^2 \widehat{R}_\hbar.\end{equation}
Let us now make a few additional comments using the fact that $A$ and $V$ belong to some space $\mathcal{A}_s$. First of all, according to 
Lemma~\ref{average_analytic_norm}, we know that, as soon as $A$ and $V$ belongs to the space $\mathcal{A}_s$, both $\langle A\rangle$ 
and $\langle V\rangle$ belong\footnote{Recall also that $\mathcal{A}_s\subset S(1)$.} to the space $\mathcal{A}_s$. Moreover 
the functions $F_1$ and $F_2$ used to define $\mathcal{F}_{\hbar}$ are constructed from $A$ and $V$ using~\eqref{solution-cohomological}.
In particular, by~\eqref{diophantine} and for every $0<\sigma < \rho$, the following inequalites hold:
$$\Vert F_1 \Vert_s \leq \Vert F_1 \Vert_{\rho - \sigma,s}\lesssim_{\rho,s}\Vert 
V\Vert_{\rho,s},\quad\text{and}\quad\Vert F_2 \Vert_s \leq \Vert F_2 \Vert_{\rho - \sigma,s}\lesssim_{\rho,s}\Vert A\Vert_{\rho,s}.$$

We can make use of this regularity information to analyse the regularity of the remainder term $\widehat{R}_\hbar$ in~\eqref{e:averaged-op-analytic}. Recall that part of this 
term comes from the remainder term when we apply the composition formula to $[\Op_{\hbar}^w(A),\Op_{\hbar}^w(F_j)]$ and to $[\Op_{\hbar}^w(V),\Op_{\hbar}^w(F_j)]$ 
for $j=1,2$. In that case, Lemma~\ref{moyal_minus_poisson} from the appendix tells us that the remainder is a pseudodifferential operator whose symbol belongs 
to $\mathcal{A}_{s-\sigma}$ for every $0<\sigma<s$. There is another contribution coming from the integral term in the Taylor formula \eqref{e:taylor_expansion} with 
$\Op_{\hbar}^w(a)$ replaced by $\widehat{P}_{\hbar}$. For that term, we first make use of Lemma~\ref{moyal_minus_poisson} and of the fact that $F_j$ solve 
cohomological equations\footnote{This comment is to handle the contribution coming from $\widehat{H}_{\hbar}$.} \eqref{e:cohomological} in order to verify that the double bracket is a pseudodifferential operator whose symbol belongs to $\mathcal{A}_{s-\sigma}$ for every 
$0<\sigma<s$. Then, an application of the analytic Egorov Lemma from the appendix (point~(1) of Lemma~\ref{analytic_egorov} with $G=\hbar (F_1+iF_2)$) 
shows that this remainder term is still 
a pseudodifferential operator whose symbol now belongs to $\mathcal{A}_{s-\sigma}$ for every $0<\sigma<s$. To summarize, we have verified that 
$\widehat{R}_{\hbar}=\Op_\hbar^w(R_{\hbar})$ with $\|R_{\hbar}\|_{s-\sigma}\leq C_{s,\sigma,\rho}$ for every $0<\sigma<s$ and 
uniformly for $0<\hbar\leq \hbar_0$.

We now perform a second conjugation whose effect will be to replace $\langle A\rangle$ in~\eqref{e:averaged-op-analytic} 
by a term involving $V$. Let $F_3$ be some real valued element 
in $\mathcal{A}_{s-\sigma}$ for some $0 < \sigma < s$ verifying $\langle F_3\rangle=F_3$. We set, for $\varepsilon>0$ small enough (independent of $\hbar$),
$$
\tilde{\mathcal{F}}_\hbar(t) := e^{ \frac{t}{\hbar} \widehat{F}_{3,\hbar}}, \quad t \in [-\varepsilon,\varepsilon],
$$
where $\widehat{F}_{3,\hbar}=\Op_{\hbar}^w(\langle F_3\rangle).$
We can define the new conjugate of $\widehat{H}_\hbar$:
$$\tilde{\mathcal{F}}_\hbar(-\varepsilon)\mathcal{F}_\hbar\widehat{P}_{\hbar}\mathcal{F}_\hbar^{-1}\tilde{\mathcal{F}}_\hbar(\varepsilon)=\widehat{H}_{\hbar}+
\hbar\tilde{\mathcal{F}}_\hbar(-\varepsilon)\left(
\Op_{\hbar}^w\left(\langle V\rangle\right)+i
\Op_{\hbar}^w\left(\langle A\rangle\right)+\hbar \widehat{R}_\hbar\right)\tilde{\mathcal{F}}_\hbar(\varepsilon),$$
where we used that $[\widehat{H}_{\hbar},\Op_\hbar^w(\langle F_3\rangle)]=0$. In fact, as $H$ is quadratic in $(x,\xi)$ and as we used the Weyl-quantization, 
the fact that $H$ and $\langle F_3\rangle$ (Poisson-)commute implies that $[\widehat{H}_{\hbar},\Op_\hbar^w(\langle F_3\rangle)]=0$. 
Suppose now that  $\varepsilon\|\langle F_3\rangle\|_{s-\sigma}\leq\frac{\sigma^2}{2}$ so that we can use the (analytic) Egorov Lemma~\ref{analytic_egorov} with $G=iF_3$. 
This tells us that \begin{equation}\label{e:remainder-analytic}\tilde{\mathcal{F}}_\hbar(-\varepsilon)\widehat{R}_\hbar\tilde{\mathcal{F}}_\hbar(\varepsilon)=
\Op_{\hbar}^w(R_{\hbar}(\varepsilon)),
\end{equation}
with $R_{\hbar}(\varepsilon)$ belonging to $\mathcal{A}_{s-\sigma}$ uniformly for $\hbar$ small enough.
Using the conventions of Appendix~\ref{a:analytic}, one also has
\begin{equation}\label{e:leading-analytic}\tilde{\mathcal{F}}_\hbar(\varepsilon)\left(
\Op_{\hbar}^w\left(\langle V\rangle\right)+i
\Op_{\hbar}^w\left(\langle A\rangle\right)\right)\tilde{\mathcal{F}}_\hbar(-\varepsilon)=\Op_{\hbar}^w\left(\Psi_{\varepsilon}^{iF_3,\hbar}(\langle V\rangle+i\langle A\rangle) \right).\end{equation}

Consider now a sequence $(\lambda_{\hbar}=\alpha_{\hbar}+i\hbar\beta_\hbar)_{0<\hbar\leq 1}$ solving~\eqref{e:eigenvalue-equation} 
with $r_{\hbar}=0$ and $\beta_\hbar\rightarrow\beta$. In particular, one can find a sequence of 
normalized eigenvectors $(\tilde{v}_{\hbar})_{0<\hbar\leq 1}$ such that
$$\tilde{\mathcal{F}}_\hbar(-\varepsilon)\mathcal{F}_\hbar\widehat{P}_{\hbar}\mathcal{F}_\hbar^{-1}\tilde{\mathcal{F}}_\hbar(\varepsilon)\tilde{v}_{\hbar}=
\lambda_{\hbar}\tilde{v}_{\hbar}.$$
Implementing~\eqref{e:remainder-analytic} and~\eqref{e:leading-analytic}, one obtains
$$\text{Im}\left\langle\Op_{\hbar}^w\left(\Psi_{\varepsilon}^{iF_3,\hbar}(\langle V\rangle+i\langle A\rangle)\right)\tilde{v}_{\hbar},\tilde{v}_{\hbar}\right\rangle
+O(\hbar)=\frac{1}{\hbar}\text{Im}\left\langle\tilde{\mathcal{F}}_\hbar(-\varepsilon)\mathcal{F}_\hbar\widehat{P}_{\hbar}\mathcal{F}_\hbar^{-1}\tilde{\mathcal{F}}_\hbar(\varepsilon)\tilde{v}_{\hbar}
,\tilde{v}_{\hbar}\right\rangle=\beta_{\hbar}.$$
From point~$(3)$ of Lemma~\ref{analytic_egorov}, one then finds
$$\beta_{\hbar}=\left\langle\Op_{\hbar}^w\left(\langle A\rangle-\varepsilon\{ \langle F_3\rangle,\langle V\rangle\}\right)\tilde{v}_{\hbar},\tilde{v}_{\hbar}\right\rangle
+O(\varepsilon^2)+O(\hbar).$$
Up to another extraction, we 
can suppose that the sequence $(\tilde{v}_{\hbar})_{\hbar>0}$ has an unique semiclassical measure $\tilde{\mu}$ which is still 
a probability measure carried by $H^{-1}(1)$. Letting $\hbar\rightarrow 0^+$, one finds
$$\beta=\tilde{\mu}\left(\langle A\rangle+\varepsilon\{ \langle V\rangle,\langle F_3\rangle\}\right)+O(\varepsilon^2).$$
Given $0 < \sigma < s$, suppose now that we can pick $F_3$ in $\mathcal{A}_{s-\sigma}$ 
such that $\{\langle F_3\rangle,\langle V\rangle\}<0$ on $\langle A\rangle^{-1}(0)\cap H^{-1}(1)$. Then, one can find some $c_0>0$  such that $c_0\varepsilon+O(\varepsilon^2)\leq\beta$. In particular, $\beta$ cannot be 
taken equal to $0$ 
which concludes the proof of Theorem~\ref{t:main_result_3} except for the proof of the existence of $F_3$.

Let us now show that the geometric control assumption~\eqref{e:control-subprincipal} of Theorem~\ref{t:main_result_3} implies 
the existence of $F_3$. Since $\langle A \rangle$ and $\langle V \rangle$ belong to $\mathcal{A}_s$, 
Remark~\ref{r:classic_remark} from the Appendix and the compactness 
of the set $H^{-1}(1) \cap \langle A \rangle^{-1}(0)$ show that, for every $0<\sigma<s$, there exists some small enough $t_0>0$ such that
$$F_3(z):=\int_0^{t_0}\left(\int_0^t\langle A\rangle\circ \phi^{\langle V\rangle}_{\tau}(z)d\tau\right)dt$$
belongs to $\mathcal{A}_{s-\sigma}.$ One has, for every $z\in H^{-1}(1)\cap\langle A\rangle^{-1}(0)$,
$$\left\{\langle V\rangle,F_3\right\}(z)=
\int_0^{t_0}\langle A\rangle\circ \phi^{\langle V\rangle}_{t}(z)dt.$$
It remains to verify that this quantity is positive for every $z_0$ in $\langle A\rangle^{-1}(0)\cap H^{-1}(1)$. Still 
using Remark~\ref{r:classic_remark}, one has the following analytic expansion:
\begin{equation}\label{e:analytic_expansion}
\langle A\rangle\circ \phi^{\langle V\rangle}_{t}(z) = \sum_{j=0}^\infty \frac{t^j}{j!} \Ad_{\langle V \rangle}^j\big( \langle A \rangle \big)(z), 
\end{equation}
uniformly for $t \in [-t_0, t_0]$ and $z \in H^{-1}(1)$. This implies that, if we fix some $z_0$ in $H^{-1}(1)$, then the map 
$t\mapsto\langle A\rangle\circ \phi^{\langle V\rangle}_{t}(z)$ is analytic on $\R$. Now, given some $z_0\in
\langle A\rangle^{-1}(0)\cap H^{-1}(1)$, there exists some $z_1$ in the orbit of $z_0$ such that 
$\langle A\rangle(z_1)>0$ thanks to our geometric control assumption\eqref{e:control-subprincipal}. In particular, the analytic map $t\mapsto\langle A\rangle\circ \phi^{\langle V\rangle}_{t}(z_0)$ 
is nonconstant and there exists some $j\geq 1$ such that 
$ \Ad_{\langle V \rangle}^j\big( \langle A \rangle \big)(z_0)\neq 0$. 
Hence, $\left\{\langle V\rangle,F_3\right\}(z_0)>0$ which concludes the proof.

\appendix

\section{Symbolic Calculus on the spaces $ \mathcal{A}_s$}\label{a:analytic}

We collect some basic lemmas about the quantization of the spaces $\mathcal{A}_s$. We fix $s>0$ all along this appendix. 
Let $a,b \in \mathcal{A}_s$. The operator given by the composition $\Op_\hbar^w(a) \Op_\hbar^w(b)$ 
is another pseudodifferential operator with symbol $c$ given by the Moyal product $c = a \sharp_\hbar b$, which can be 
written by the following integral formula  \cite[Ch.~7, p.~79]{Dim_Sjo99}:
\begin{equation}
\label{moyal_product}
c(z) = a \sharp_\hbar b(z) = \frac{1}{(2\pi)^{4d}} \int_{\R^{4d}} \widehat{a}(w^*)\widehat{b}(z^*-w^*) e^{\frac{i \hbar}{2} \varsigma(w^*,z^*-w^*)} e^{i z^* \cdot z} dw^* \, dz^*,
\end{equation}
where $\varsigma(x,\xi,y,\eta) := \xi \cdot y - x \cdot \eta$ is the standard symplectic product and where
$$\widehat{a}(w):=\int_{\mathbb{R}^{2d}}e^{-iw \cdot z}a(z)dz.$$
We set $[a,b]_\hbar := a \sharp_\hbar b - b \sharp_\hbar a.$ Given now $ a, G \in \mathcal{A}_s$, the following conjugation formula holds formally:
$$
e^{i\frac{t}{\hbar} \Op_\hbar^w(G)} \Op_\hbar^w(a) e^{-i\frac{t}{\hbar} \Op_\hbar^w(G)} = \Op_\hbar^w(\Psi_t^{G,\hbar} a),
$$
where
\begin{equation}
\label{analytic_conjugation}
\Psi_t^{G,\hbar} a := \sum_{j=0}^\infty \frac{1}{j!} \left( \frac{it}{\hbar} \right)^j \Ad_G^{\sharp_\hbar,j}(a), \quad t \in \R,
\end{equation}
and
$$
\Ad_G^{\sharp_\hbar,j}(a) =  [G, \Ad_G^{\sharp_\hbar,j-1}(a)]_\hbar, \quad \Ad_G^{\sharp_\hbar,0}(a) = a.
$$
One of the aim of this appendix is to prove the following analytic version of Egorov's theorem:

\begin{lemma}[Analytic Egorov's Lemma]
\label{analytic_egorov}
Let $0 < \sigma < s$. Consider the family of Fourier integral operators $\{ \mathcal{G}_\hbar(t) : t \in \R \}$ defined by
$$
\mathcal{G}_\hbar(t) := e^{-\frac{it}{\hbar}\widehat{G}_{\hbar} }, 
$$
where $\widehat{G}_{\hbar} = \Op_\hbar^w(G)$ for some $G\in\mathcal{A}_s$. Assume
\begin{equation}
\label{c:small_condition}
  |t|<\frac{ \sigma^2}{ 2\Vert G \Vert_s }.
\end{equation}
Then, there exists a constant $C_{\sigma}>0$ (depending only on $\sigma$) such that, for every $a \in \mathcal{A}_s$,
\medskip

\begin{enumerate}
\item  $\Psi_t^{G,\hbar} a \in \mathcal{A}_{s-\sigma}$;
\medskip

\item 
$\displaystyle\big \Vert \Psi_t^{G,\hbar}a - a \big \Vert_{s-\sigma} \leq  C_{\sigma}\vert t \vert \Vert G \Vert_s 
  \Vert a \Vert_{s}$;

\medskip

\item $\Vert \Psi_t^{G,\hbar} a - a+t\{G,a\}\Vert_{s-\sigma}\leq C_{\sigma}|t|^2\|G\|_s\|a\|_s$ for some $C_{\sigma}>0$ depending only on $\sigma$. 
\end{enumerate}
\end{lemma}

\begin{remark}
\label{r:classic_remark}
With the hypothesis of Lemma \ref{analytic_egorov}, one also 
has that $a \circ \phi_t^G \in \mathcal{A}_{s-\sigma}$. To see this, 
it is enough to follow verbatim the proof of Lemma \ref{analytic_egorov} 
noting that Lemma \ref{analytic_conmutator} below remains valid for $-i\hbar \{a,b \}$ 
instead of $[a,b]_\hbar$ and then using the formal expansion
$$
a \circ \phi_t^G = \sum_{j=0}^\infty \frac{t}{j!} \Ad_G^j(a), 
$$
where $\Ad_G^j(a) = \{ G, \Ad_G^{j-1}(a) \}$ and $\Ad_G^0(a) = a$ instead of the analogous quantities for $\Psi_t^{G,\hbar}a$.
\end{remark}

\subsection{Preliminary lemmas}

Before proceeding to the proof, we start with some preli\-minary results.

\begin{lemma}
\label{l:multiplicative_property}
For every $a, b \in \mathcal{A}_s$, the following holds:
$$
\Vert a b \Vert_s \leq \Vert a \Vert_s \Vert b \Vert_s.
$$
\end{lemma}
\begin{proof}
To see this, write
\begin{align*}
\Vert a b \Vert_s & = \int_{\R^{2d}} \vert \widehat{ ab }(w) \vert e^{s \vert w \vert} dw \\[0.2cm]
 & = \int_{\R^{2d}} \left \vert \int_{\R^{2d}} \widehat{a}(w - w^*) \widehat{b}(w^*) dw^* \right \vert e^{s \vert w \vert} dw \\[0.2cm]
 & \leq \int_{\R^{2d}} \int_{\R^{2d}} \vert \widehat{a}(w-w^*) \vert e^{s \vert w - w^* \vert} \vert \widehat{b} (w^*) \vert e^{s \vert w^* \vert} dw^*\, dw \\[0.2cm]
 &\leq \Vert a \Vert_s \Vert b \Vert_s.
\end{align*}
\end{proof}

We shall also need some estimates on the Moyal product of elements in $\mathcal{A}_s$:
\begin{lemma}
\label{analytic_conmutator}
Let $a,b \in \mathcal{A}_s$. Then, for every $ 0 < \sigma_1 + \sigma_2 < s$, $[a,b]_\hbar \in \mathcal{A}_{s-\sigma_1-\sigma_2}$ and
$$
\left\Vert [a,b]_\hbar \right\Vert_{s-\sigma_1 - \sigma_2} \leq \frac{2\hbar}{ e^2 \sigma_1 (\sigma_1 + \sigma_2)} \Vert a \Vert_s \Vert b \Vert_{s-\sigma_2}.
$$
\end{lemma}

\begin{proof}
From (\ref{moyal_product}), we have
$$
[a,b]_\hbar(z) = 2i \int_{\R^{4d}}  \widehat{a}(w^*)\widehat{b}(z^*-w^*) \sin \left( \frac{\hbar}{2} \, \varsigma(w^*,z^*-w^*) \right) \frac{e^{i z^* \cdot z}}{(2\pi)^{4d}} dw^* \, dz^*.
$$
Then, using that
\begin{equation}
\label{e:estimate_simplectic_product}
\vert \varsigma(w^*,z^*-w^*) \vert \leq 2 \vert w^* \vert \vert z^* - w^* \vert,
\end{equation}
we obtain:
\begin{align*}
\Vert [a,b]_\hbar \Vert_{s- \sigma_1 - \sigma_2} & \\[0.2cm]
 & \hspace*{-1cm} \leq \frac{2\hbar}{(2\pi)^{4d}} \int_{\R^{4d}} \vert \widehat{a}(w^*)\vert \vert w^* \vert \vert \widehat{b}(z^*-w^*) \vert \vert z^* - w^* \vert e^{(s-\sigma_1 - \sigma_2)( \vert z^* - w^* \vert + \vert w^* \vert)}  dw^* \, dz^* \\[0.2cm]
 &  \hspace*{-1cm}  \leq \frac{2\hbar}{(2\pi)^{4d}} \big( \sup_{r \geq 0} r e^{-\sigma_1 r} \big)\big( \sup_{r \geq 0} r e^{-(\sigma_1+\sigma_2) r} \big) \Vert a \Vert_s \Vert b \Vert_{s-\sigma_2} \\[0.2cm]
 &  \hspace*{-1cm} \leq \frac{2\hbar}{ e^2 \sigma_1(\sigma_1 + \sigma_2)} \Vert a \Vert_s \Vert b \Vert_{s-\sigma_2}.
\end{align*}
\end{proof}

Finally, one has:

\begin{lemma}
\label{moyal_minus_poisson}
Let $a,b \in \mathcal{A}_{s}$ and $0 < \sigma < s$. Then there exists a contant $C_\sigma > 0$ depending only on $\sigma$ such that
 
\begin{equation}
\label{e:commutator_poisson}
\left\Vert \frac{i}{\hbar} [a,b]_\hbar - \{ a, b \} \right\Vert_{s-\sigma} \leq C_\sigma\,  \hbar^2 \Vert a \Vert_{s} \Vert b \Vert_{s-\sigma}.
\end{equation}

\end{lemma}

\begin{proof}
First write:
\begin{align*}
\hspace*{0.1cm} [a,b]_\hbar(z) + i \hbar \{ a, b \}(z) & \\[0.2cm]
 & \hspace*{-3.4cm} = 2i \int_{\R^{4d}}  \widehat{a}(w^*)\widehat{b}(z^*-w^*) \left( \sin \left( \frac{\hbar}{2} \, \varsigma(w^*,z^*-w^*) \right) - \frac{\hbar}{2} \, \varsigma(w^*,z^*-w^*) \right) \frac{e^{i z^* \cdot z}}{(2\pi)^{4d}} dw^* \, dz^*.
\end{align*}
Using~\eqref{e:estimate_simplectic_product} and $\sin(x) = x - \frac{x^2}{2} \int_0^1 \sin(tx) (1-t) dt$,
we obtain
\begin{align*}
\Vert [a,b]_\hbar + i \hbar \{ a, b \} \Vert_{s-\sigma} &  \\[0.2cm]
 & \hspace*{-2cm} \leq \frac{\hbar^3}{(2\pi)^{4d}} \int_{\R^{4d}} \vert \widehat{a}(w^*)\vert \vert w^* \vert^3 \vert \widehat{b}(z^*-w^*) \vert \vert z^* - w^* \vert^3 e^{(s-\sigma)( \vert z^* - w^* \vert + \vert w^* \vert)}  dw^* \, dz^* \\[0.2cm]
 &  \hspace*{-2cm}  \leq C_\sigma \, \hbar^3 \Vert a \Vert_s \Vert b \Vert_{s-\sigma}.
\end{align*}

\end{proof}

\subsection{Proof of the analytic Egorov Lemma}
We are now in position to prove Lemma~\ref{analytic_egorov}. Let us start with points~$(1)$ and~$(2)$.
By definition (\ref{analytic_conjugation}), we have
$$
\Vert  \Psi_t^{G,\hbar} a - a \Vert_{s-\sigma} \leq 
\sum_{j=1}^\infty \frac{1}{j!} \left( \frac{\vert t \vert}{\hbar} \right)^j \Vert \Ad_G^{\sharp_\hbar,j}(a) \Vert_{s - \sigma}.
$$
Using Lemma \ref{analytic_conmutator}, we also find that, for every $j \geq 1$,
\begin{align*}
\Vert \Ad_G^{\sharp_\hbar,j}(a) \Vert_{s - \sigma} & \leq \frac{2\hbar j}{e^2\sigma^2} \Vert \Ad_G^{\sharp_\hbar,j-1}(a) 
\Vert_{s - \frac{(j-1)\sigma}{j}} \Vert G \Vert_s \\[0.2cm]
 & \leq \frac{2^2\hbar^2 j^3}{e^4 \sigma^4(j-1) } \Vert  \Ad_G^{\sharp_\hbar,j-2}(a) \Vert_{s - \frac{(j-2)\sigma}{j}} \Vert G \Vert_s^2 \\[0.2cm]
 & \leq \cdots \leq \frac{2^j\hbar^j j^{2j}}{e^{2j} \sigma^{2j}j!} \Vert a \Vert_s \Vert G \Vert_s^j.
\end{align*} 
Then, using Stirling formula and as $\frac{2\vert t \vert \Vert G \Vert_{s} }{ \sigma^2}  < 1$, one gets
\begin{equation}\label{e:leading-term-analytic-egorov}
\Vert  \Psi_t^{G,\hbar} a - a \Vert_{s-\sigma} \leq  
\sum_{j=1}^\infty \frac{ j^{2j}\vert t \vert^j \Vert G \Vert_s^j}{(j!)^2(e \sigma)^{2j}}  \Vert a \Vert_s \leq  C_{\sigma}|t| \|G\|_s  \Vert a \Vert_s,
\end{equation}
for some constant $C_{\sigma}>0$ depending only on $\sigma$. In order to prove point~(3), we now write
\begin{eqnarray*}\Vert \Psi_t^{G,\hbar} a - a+t\{G,a\}\Vert_{s-\sigma} &\leq & 
|t|\left\Vert \frac{i }{\hbar}[G,a]_{\hbar}-\{G,a\}\right\Vert_{s-\sigma} \\
&+&\sum_{j=2}^\infty \frac{1}{j!} \left( \frac{\vert t \vert}{\hbar} \right)^j \Vert \Ad_{G}^{\sharp_\hbar,j}(a) \Vert_{s - \sigma}.
\end{eqnarray*}
We can now reproduce the above argument and combining this bound to Lemma~\ref{moyal_minus_poisson}, we can deduce point (3) of Lemma~\ref{analytic_egorov}.

\bibliographystyle{plain}
\bibliography{Referencias}
\end{document}